\documentclass[12pt]{iopart}

\usepackage[T1]{fontenc}

\usepackage{graphicx}
\usepackage{dcolumn}
\usepackage{bm}
\usepackage{comment}
\usepackage[colorlinks=true]{hyperref}

\expandafter\let\csname equation*\endcsname\relax
\expandafter\let\csname endequation*\endcsname\relax
\usepackage{amsmath}

\usepackage{mathtools}
\allowdisplaybreaks

\usepackage{amsthm}
\newtheorem{theorem}{Theorem}

\newtheorem{lemma}[theorem]{Lemma}

\newtheorem{proposition}[theorem]{Proposition}

\begin{document}

\title{Provable Optimality of the Square-Tooth Atomic Frequency Comb Quantum Memory}

\author{Allen Zang}
\address{Pritzker School of Molecular Engineering, University of Chicago, Chicago, IL, USA}
\ead{yzang@uchicago.edu}

\author{Martin Suchara}
\address{Microsoft Azure Quantum, Microsoft Corporation, Redmond, WA, USA}

\author{Tian Zhong}
\address{Pritzker School of Molecular Engineering, University of Chicago, Chicago, IL, USA}

\date{\today}

\begin{abstract}
    Atomic frequency comb (AFC) quantum memories are a promising technology for quantum repeater networks because they enable multi-mode, long-time, and high-fidelity storage of photons with on-demand retrieval. The optimization of the retrieval efficiency of an AFC memory is important because it strongly impacts the entanglement distribution rate in quantum networks. Despite initial theoretical analyses and recent experimental demonstrations, a rigorous proof of the universally optimal configuration for the highest AFC retrieval efficiency has not been presented. In this paper we present a simple analytical proof which shows that the optimized square tooth offers the highest retrieval efficiency among all tooth shapes, under the physical constraint of finite optical depth of an atomic ensemble. The optimality still holds when the non-zero background absorption and the finite optical linewidth of atoms are considered. We further compare square, Lorentzian and Gaussian tooth shapes to reinforce the practical advantage of the square-tooth AFC in retrieval efficiency. Our proof lays rigorous foundation for the recipe of creating optimal AFC under realistic experimental conditions. 
\end{abstract}

\maketitle

\section{Introduction}
Quantum memories capable of storing quantum states for long periods of time are essential components of quantum communication networks~\cite{kimble2008quantum,wehner2018quantum}. Quantum networks are receiving significant attention in the science and engineering community because they are expected to enable new important applications such as distributed quantum computing~\cite{gottesman1999demonstrating,cuomo2020towards}, distributed sensing~\cite{proctor2018multiparameter,zhang2021distributed}, and secure key distribution~\cite{bennett2014quantum,ekert1991quantum}. 

Atomic frequency comb (AFC)~\cite{de2008solid,afzelius2009multimode} is a promising optical quantum memory protocol~\cite{lvovsky2009optical,lei2023quantum} which allows the absorption and retrieval of photons transmitted within quantum networks~\cite{awschalom2021development,awschalom2022roadmap}, and enables entanglement distribution between remote network nodes~\cite{briegel1998quantum,guha2015rate,munro2015inside,muralidharan2016optimal,pant2019routing,dahlberg2019link,shi2020concurrent,khatri2021policies,goodenough2021optimizing,kolar2022adaptive,zang2023entanglement,azuma2023quantum,zang2024analytical} when combined with common single- or entangled-photon sources~\cite{sangouard2011quantum}. Incoming optical signals to AFC memories are stored in a delocalized form within the atomic ensemble. The retrieval of stored photons in AFC memories is achieved due to the comb-like density function of atomic transition frequency, or in classical terms, the comb-like absorption profile in the frequency domain. This results in the rephasing of all emitters' optical transition dipoles after a fixed storage time determined by the comb period, leading to re-emission of the stored photons. Furthermore, the retrieval from the AFC memory can be achieved on-demand by introducing control pulses that convert atomic ensemble excitation between the optical transition and the long-lived spin-wave~\cite{afzelius2010demonstration,lauritzen2011approaches}. 

AFC memories have one significant advantage over many other absorptive quantum memories such as electromagnetically induced transparency (EIT)~\cite{fleischhauer2000dark,phillips2001storage} and Raman storage~\cite{michelberger2015interfacing,saunders2016cavity}. The temporal multimodality of AFC memories~\cite{ortu2022multimode} is in principle independent of optical depth~\cite{afzelius2009multimode}, but instead is determined by the storage time and the duration of the signal to store, and more specifically the ratio of inhomogeneous broadening to homogeneous broadening. We note that the temporal multimodality of other photon-echo quantum memory protocols~\cite{chaneliere2018quantum,moiseev2024echo} are also less limited by optical depth, such as controlled inhomogeneous broadening (CRIB)~\cite{moiseev2001complete,moiseev2004possibilities,alexander2006photon,sangouard2007analysis}, gradient echo memory (GEM)~\cite{alexander2006photon}, and revival of silenced echo signal (ROSE)~\cite{damon2011revival}, while there is still evidence that AFC has a better support for multimodality under identical optical depth~\cite{nunn2008multimode}. Multi-mode quantum memories are important for quantum network architectures due to the requirement of  multiplexing~\cite{collins2007multiplexed,simon2007quantum,sangouard2011quantum}. Moreover, the multimodality of AFC memories goes beyond temporal modes, and extends to spatial and spectral degrees of freedom. These modes can be combined, as has been reported in experiments~\cite{gundougan2012quantum,tang2015storage,jobez2016towards,yang2018multiplexed,seri2019quantum,businger2022non,wei2024quantum}. Notably, the distribution of entanglement between AFC quantum memories has been experimentally demonstrated~\cite{lago2021telecom,liu2021heralded}. Furthermore, one-hour coherent storage of optical signal via a spin-wave AFC memory has been realized with dynamical decoupling for noise-mitigation~\cite{ma2021one}. Due to the promising potential of their integration into quantum networks, models of AFC memories have been included in quantum network simulators that model the quantum network physical layer, such as NetSquid~\cite{coopmans2021netsquid} and SeQUeNCe~\cite{wu2021sequence,zang2022simulation}. 

The retrieval efficiency of AFC memories is a potential limiting factor of the entanglement generation rate for distributed quantum information processing. For AFC memories without cavity enhancement as considered in this work, the upper bound of forward retrieval efficiency is 54\% and the upper bound of backward retrieval efficiency is 100\%~\cite{afzelius2009multimode}. However, such theoretical upper bounds do not explicitly instruct experimentalists how to achieve the best possible performance under realistic constraints. The AFC tooth shape optimization was first reported more than a decade ago~\cite{bonarota2010efficiency,bensky2012optimizing}, where the results imply that AFC memories with square teeth could achieve the highest retrieval efficiency. However, to the best of our knowledge, no rigorous proof of the optimality for square-tooth AFC has been presented since. 

In this work we use the semi-classical theoretical framework from~\cite{bonarota2010efficiency} to rigorously prove that AFCs with square teeth provide the highest achievable retrieval efficiency when the height of comb teeth is upper bounded, i.e., when the optical depth is finite. We show that this optimality still holds when the spectral profile has non-zero background absorption and the intrinsic line shape of an atom has finite width. The paper is organized as follows. Our proof covers all integrable functions which arguably include all physically relevant tooth shapes. Therefore, this work reinforces physical intuition with mathematical rigor by eliminating the existence of any potential corner cases. Moreover, the proof itself could inspire other studies of optimality. In Section~\ref{sec:review} we briefly review the semi-classical theory describing the physical models of AFC memories and the analytical expression of forward retrieval efficiency, as shown in~\cite{bonarota2010efficiency}. We then prove the optimality of the retrieval efficiency for AFCs with square teeth in Section~\ref{sec:proof}, both in the ideal case and when considering non-zero background absorption and optical linewidth. After the proofs, in Section~\ref{sec:comparison} we compare the square tooth with two other typical lineshapes, namely Lorentzian and Gaussian, to demonstrate the robust advantage of square tooth. Section~\ref{sec:conclusion} concludes the paper.

\section{Semi-classical theory of the atomic frequency comb}\label{sec:review}
The periodic comb structure in the distribution of the atomic transition frequency for AFC memories allows us to interpret the absorption and retrieval processes of optical AFC memories as diffraction of a spectral grating~\cite{sonajalg1994diffraction,chaneliere2010efficient,bonarota2010efficiency} under a semi-classical theoretical framework. Therefore, the analytical expression of AFC retrieval efficiency can be derived from semi-classical Maxwell-Bloch equations~\cite{arecchi1965theory,icsevgi1969propagation}, which describe the absorption and retrieval processes of AFC by incorporating the coupled dynamics of both the propagating electromagnetic field and the two-level system ensemble in the medium. 

The first-order Maxwell-Bloch equations are derived under the slowly varying envelope approximation and the rotating wave approximation~\cite{allen1987optical}, and in the weak input signal limit~\cite{crisp1970propagation}. The two equations are:
\begin{equation}\label{eqn:MB_eqn}
\begin{aligned}
    &\partial_z \Omega(z,t) + \frac{1}{c}\partial_t\Omega(z,t) = -\frac{i}{2\pi}\int d\omega f(\omega)P(\omega;z,t),\\
    &\partial_t P(\omega;z,t) = -(i\omega + \gamma)P(\omega;z,t) - i\Omega(z,t),
\end{aligned}
\end{equation}
where $\Omega(z,t)$ is the Rabi frequency proportional to the propagating field in the medium, $P(\omega;z,t)$ is the atomic polarization for two-level systems with frequency detuning $\omega$, $f(\omega)$ is the frequency-dependent absorption coefficient which represents the effect of inhomogeneous broadening~\cite{sangouard2007analysis}, and the effect of homogeneous broadening is accounted phenomenologically by the parameter $\gamma$~\cite{icsevgi1969propagation}. We further assume a periodic absorption coefficient 
\begin{align}
    f(\omega) = \sum_{n\leq 0}F_n e^{-in\omega T}
\end{align}
to account for the comb structure, where $2\pi/T$ is the AFC comb period, and the requirement of $n\leq 0$ is for a causality reason~\cite{chaneliere2010efficient}. The periodic structure also implies successive retrieval ``echoes'' centered at times $t=pT,\ p=0,1,2,\dots$, i.e. 
\begin{align}
    \Omega(z,t)=\sum_{p\geq 0}a_p(z)\Omega(0,t-pT),
\end{align}
where the requirement on $p\geq 0$ comes again from causality. 

The retrieval efficiency is defined as $\eta(L)\coloneqq|a_1(L)|^2$ corresponding to the first retrieved pulse at the output end of the sample with length $L$. While the detailed solution of Eqn.~\ref{eqn:MB_eqn} is presented in Appendix~\ref{sec:solution}, we note that from the solution we have the analytical expression for the AFC retrieval efficiency:
\begin{equation}\label{eqn:eff}
    \eta(L) = |F_{-1}L|^2e^{-F_0L},
\end{equation}
where $F_0$ and $F_{-1}$ are the 0-th and $(-1)$-th Fourier coefficients, respectively, for a periodic function $f(\omega)$ which defines the shape of the periodic comb. Explicitly they are calculated as:
\begin{align}
    & F_0 = \frac{T}{2\pi}\int_{-\pi/T}^{\pi/T}f(\omega)d\omega, \\
    & F_{-1} = \frac{T}{2\pi}\int_{-\pi/T}^{\pi/T}f(\omega)e^{i\omega T}d\omega.
\end{align}

\section{Proof of optimality of the square-tooth atomic frequency comb}\label{sec:proof}
For the rest of the paper we will focus on the mathematical proofs, before which we would like to emphasize the connection between the abstract mathematical formulation and the physical scenario. We consider the scenario where experimentalists perform optical pumping on a specific atomic ensemble to create AFC with a fixed periodicity in frequency domain $\delta=\pi/T$ (and the full period is $2\delta$), which is determined by the retrieval time $T$ required by the quantum information process task to achieve. Besides the fixed period, the created AFC tooth shape will also be subject to constraints which are determined by the physical properties of the atomic ensemble itself: maximum absorption $\alpha_M$, background absorption $\alpha_\mathrm{bg}$, and optical linewidth from finite optical coherence time $T_2$ together with spectral diffusion and broadening due to the hole-burning laser. Under the assumption of periodic comb, we can focus on a single period $[-\pi/T,\pi/T]$ (i.e. $[-\delta,\delta]$), in which the shape of comb tooth is described by a real-valued function $f(\omega)$ s.t. $\alpha_\mathrm{bg}\leq f(\omega)\leq \alpha_M$.

In the following, we first ignore background absorption, i.e. assuming $\alpha_\mathrm{bg}=0$, and prove the optimality of square tooth among all symmetric tooth shapes s.t. $f(\omega)=f(-\omega)$, which will be the basis of proving the optimality of square tooth among arbitrary tooth shapes after coordinate redefinition. Then, we show that the inclusion of non-zero background absorption does not affect the optimality of square tooth above the background. Finally, we consider finite optical linewidth which will make it impossible to create arbitrary comb tooth shapes. We demonstrate that taking the optimized square-tooth as the \textit{target} tooth shape will give the best efficiency among all realizable tooth shapes, under the assumption that the actual tooth shape is the convolution of the target tooth shape and the single-atom line shape~\cite{bonarota2010efficiency}.

\subsection{Proof of optimality among all symmetric tooth shapes}\label{sec:symmetric}
We first consider symmetric tooth shapes s.t. $f(\omega)=f(-\omega)$, because they allow us to ignore the imaginary part of $F_{-1}$ which naturally vanishes due to the symmetry of $f(\omega)$, and we can also restrict ourselves to positive $\omega$. Notice that tooth shapes subject to the maximal absorption constraint $\alpha_M$ can have varying area $F_0\in[0,\alpha_M]$. Therefore, the proof of the optimality of the \textit{optimized} square tooth is achieved in two steps. We first prove that for all symmetric combs with identical area (thus identical $F_0$), the square one is the best. Then we can optimize the width for all allowed square teeth, and the square tooth with optimized width is then obviously the best among all possible tooth shapes even with different areas.

\begin{proposition}\label{thm:sym_comb}
Among all symmetric tooth shapes defined on $[-\pi/T,\pi/T]$ with fixed area and subject to $0\leq f(\omega)\leq\alpha_M$, the square tooth provides the highest retrieval efficiency.
\end{proposition}
Here we explain the sketch of the proof, while the details can be found in Appendix~\ref{sec:proof_detail}. The square tooth shape $f_s(\omega)$ takes value $\alpha_M$, i.e. the maximal absorption, on $[-\Gamma,\Gamma]$, where $\Gamma$ is the half-width of square tooth, and is zero on $[-\pi/T,-\Gamma]\cup[\Gamma,\pi/T]$. We then consider an arbitrary symmetric shape $f(\omega)$ with the same area as the square tooth, which means that $f(\omega)$ satisfies:
\begin{align}
    \int_0^{\pi/T}f(\omega)d\omega = \int_0^{\Gamma}\alpha_Md\omega = \alpha_M\Gamma.
\end{align}
To prove the proposition, we would like to show that $f(\omega)$ cannot achieve higher retrieval efficiency than the square shape $f_s(\omega)$, i.e.
\begin{align}
    \left\vert\int_0^{\pi/T}f(\omega)\cos(\omega T)d\omega\right\vert \leq \int_0^{\Gamma}\alpha_M\cos(\omega T)d\omega,
\end{align}
which is then proved based on bounding of the integrands on the interval $[0,\pi/T]$. The proof is divided into two scenarios: (i) $\Gamma\leq\pi/2T$, i.e. when the finesse of the comb is greater than 2, and (ii) $\Gamma>\pi/2T$, i.e. when the finesse of the comb is smaller than 2.

This proposition means that we cannot find another symmetric tooth shape $f(\omega)$ that can achieve a higher retrieval efficiency than the square shape when they have the same area. In other words, for any possible symmetric tooth shape $f(\omega)$, there at least exists one square tooth which can achieve an efficiency that is not lower than $f(\omega)$, and this square tooth has the same area as $f(\omega)$. As a result, the symmetric tooth shape which can achieve the highest efficiency must be a square tooth, and then we can straightforwardly optimize the width of square tooth to obtain the optimal square tooth~\cite{bonarota2010efficiency}, which offers the highest efficiency among all possible symmetric combs subject to the constraint $0\leq f(x)\leq \alpha_M$.

Mathematically, the above proposition means that for any bounded function $0\leq f(\omega)\leq\alpha_M$ defined on $[0,\pi/T]$ with area $\int_0^{\pi/T}f(\omega)d\omega = \alpha_M\Gamma$, we always have $\left\vert\int_0^{\pi/T}f(\omega)\cos(\omega T)d\omega\right\vert \leq \int_0^{\Gamma}\alpha_M\cos(\omega T)d\omega$. This will help the proof of optimality among arbitrary shapes in the next section.

\subsection{Proof of optimality among all tooth shapes}
For an arbitrary tooth shape function, the Fourier coefficient $F_{-1}$ will in general contain imaginary part. According to Eqn.~\ref{eqn:eff} we are interested in the modulus of $F_{-1}$. 

Geometrically, for $\omega\in[-\pi/T,\pi/T]$ the integral kernel of $F_{-1}$, $e^{-i\omega T}$, represents a vector with unit length pointing towards the direction determined by the phase $\omega T$ with respect to the polar axis. The integral can then be understood as an addition of vectors pointing to the polar angle $\omega T$ with length $f(\omega)$, under the standard Riemann integral interpretation. It is certain that the integral will also result in a vector pointing to certain polar angle $\omega_0T$. Since we are only interested in the length of this vector to evaluate the AFC efficiency, the exact phase of $F_{-1}$ does not matter. We can redefine the origin of the polar angle as $\omega=\omega_0$. 

In this way, $|F_{-1}|$ equals the integral of the vector component that is parallel to the new zero-angle orientation, as now we are sure that the integral of the orthogonal component will vanish:
\begin{align}
    |F_{-1}| =& \frac{T}{2\pi}\left\vert\int_{-\pi/T}^{\pi/T}f(x)\cos(\omega T-\omega_0T)d\omega\right\vert \nonumber \\
    =& \frac{T}{2\pi}\left\vert\int_{-\pi/T-\omega_0}^{\pi/T-\omega_0}f(\omega'+\omega_0)\cos(\omega'T)d\omega'\right\vert\nonumber \\
    =& \frac{T}{2\pi}\left\vert\int_{-\pi/T}^{\pi/T}f'(\omega')\cos(\omega'T)d\omega'\right\vert,\label{eqn:modulus}
\end{align}
where we have defined a new shape function $f'(\omega')$ using translation of the coordinate $\omega'=\omega-\omega_0$:
\begin{equation}
    f'(\omega') = 
    \begin{cases}
    f(\omega'+\omega_0) & \omega'\in[-\frac{\pi}{T}, \frac{\pi}{T}-\omega_0]\\
    f(\omega'+\omega_0-\frac{2\pi}{T}) & \omega'\in[\frac{\pi}{T}-\omega_0, \frac{\pi}{T}]
    \end{cases}.
\end{equation}
Because of the assumption that the comb is periodic, for a single tooth the definition of $f'(\omega')$ is equivalent to left translating the original tooth shape $f(\omega)$ by $\omega_0$ under a periodic boundary at $\omega=\pm\pi/T$. Therefore, the new function $f'(\omega')$ still satisfies the constraint which $f(\omega)$ satisfies, i.e. $0\leq f'(\omega')\leq\alpha_M$, and its integral on $[-\pi/T,\pi/T]$ is identical to the original shape function, i.e. $\int_{-\pi/T}^{\pi/T}f(\omega)d\omega=\int_{-\pi/T}^{\pi/T}f'(\omega')d\omega'$. 

We emphasize that the coordinate translation $\omega_0$ is determined by the condition $\int_{-\pi/T}^{\omega_0}f(\omega)\sin[(\omega-\omega_0)T]d\omega = -\int_{\omega_0}^{\pi/T}f(\omega)\sin[(\omega-\omega_0)T]d\omega$, but this does not guarantee that the area of the original shape function $f(\omega)$ has equal areas on both sides of $\omega_0$, i.e. in general $\int_{-\pi/T}^{\omega_0}f(\omega)dx \neq \int_{\omega_0}^{\pi/T}f(\omega)d\omega$, or in terms of the newly defined function $\int_{-\pi/T}^0f'(\omega')d\omega' \neq \int_0^{\pi/T}f'(\omega')d\omega'$. Therefore, before we prove that the square tooth shape is optimal among arbitrary shapes with an identical area, we need the following lemma.
\begin{lemma}\label{thm:center_square}
Among all square shape functions $f_{(\Gamma,c)}(\omega)$ on $[-\pi/T,\pi/T]$ with fixed half width $0\leq\Gamma\leq\pi/T$ centered at $-\pi/T+\Gamma\leq c\leq\pi/T-\Gamma$ that take the maximal possible value $\alpha_M$ for $\omega\in[c-\Gamma,c+\Gamma]$ and zero otherwise, the one centered at $c=0$ will provide the highest 
\begin{equation}
    I(\Gamma,c)\coloneqq \int_{-\pi/T}^{\pi/T}f_{(\Gamma,c)}(\omega)\cos(\omega T)d\omega.
\end{equation}
\end{lemma}
The proof is straightforward, and can be found in Appendix~\ref{sec:proof_detail}. Using this lemma we can obtain the optimality of the square tooth shape among arbitrary shapes with a fixed area.

\begin{proposition}\label{thm:all_comb}
Among all tooth shapes defined on $[-\pi/T,\pi/T]$ with a fixed area and subject to $0\leq f(\omega)\leq\alpha_M$, the square tooth provides the highest retrieval efficiency.
\end{proposition}
The proof is based on the aforementioned coordinate redefintion, from which we can express $|F_{-1}|$ as the absolute value of an integral in Eqn.~\ref{eqn:modulus}. The form of integral is almost the same as the one we encounter in the proof of optimality among all symmetric shapes, and we can upper bound the absolute value by dividing the integral into two parts, on $[-\pi/T,0]$ and $[0,\pi/T]$, respectively. Therefore, we are able to use the results from the last section to prove the upper bound of $|F_{-1}|$ for arbitrary tooth shape. The details can be found in Appendix~\ref{sec:proof_detail}.

Hence, we have established that among all tooth shapes with the same area, the square tooth offers the highest retrieval efficiency. Then similar to the argument at the end of Sec.~\ref{sec:symmetric}, to obtain the globally optimal tooth shape we only need to optimize the square tooth, and the optimal half width for the square tooth can be easily obtained~\cite{bonarota2010efficiency}. We also present it explicitly in Appendix~\ref{sec:afc_eff}.

\subsection{Effect of background absorption}
The optimality naturally extends to the scenario where the tooth shape has a non-zero background absorption, i.e. the minimum value of the shape function is a constant $\alpha_\mathrm{bg}>0$.

We start by clarifying that when there exists a non-zero background absorption, the tooth shape refers to the shape above the background. We thus decompose the tooth shape function as $f(\omega) = f_\mathrm{bg}(\omega) + f_\mathrm{abg}(\omega)$, where subscript ``bg'' refers to the constant background and ``abg'' denotes above-background. Then we can write the retrieval efficiency as:
\begin{align}\label{eqn:background}
    \eta(L) = |(F_{\mathrm{bg},-1}+F_{\mathrm{abg},-1})L|^2e^{-(F_{\mathrm{bg},0}+F_{\mathrm{abg},0})L},
\end{align}
where $F_{\mathrm{bg},-1}$ and $F_{\mathrm{bg},0}$ are real constants. Then, following the previous proofs, we consider tooth shapes with identical areas, so that we only need to maximize $|(F_{\mathrm{bg},-1}+F_{\mathrm{abg},-1})|^2$. According to Proposition~\ref{thm:all_comb}, square $f_\mathrm{abg}(\omega)$ will achieve the highest $|F_{\mathrm{abg},-1}|$. Moreover, $F_{\mathrm{abg},-1}$ is real for square $f_\mathrm{abg}(\omega)$. Therefore, by the vector addition argument, it is obvious that square $f_\mathrm{abg}(\omega)$ will achieve the highest $|(F_{\mathrm{bg},-1}+F_{\mathrm{abg},-1})|$.

We comment that with a finite background optical depth $\mathrm{OD}_0=d_0L$, the tooth width needs to be optimized for a new effective optical depth $\mathrm{OD}'=\mathrm{OD}-\mathrm{OD}_0$. The effect of background in AFC absorption profile was first considered in~\cite{de2008solid}, where the authors approximate the effect of the background as a reduction factor on the expected efficiency. However, according to Eqn.~\ref{eqn:background} it is clear that the effect of a constant background may depend on shape of comb above it. Nonetheless, we manage to show that it does not affect the optimality of above-background square tooth shape.

\subsection{Effect of optical linewidth}
We have proved the retrieval efficiency optimality of the square-tooth AFC. However, in practice where the optical linewidth is always finite so the ideal square-tooth is never achievable. It has been commonly considered~\cite{bonarota2010efficiency} that the actual observable tooth shape will be the convolution of the ideal \textit{target} tooth shape $f(\omega)$ that we aim at generating and the normalized optical line shape $\mathcal{L}(\omega)$, i.e. $f(\omega)\rightarrow f(\omega)*\mathcal{L}(\omega)$. Note that the line shape should be normalized since the actual tooth shape should be identical to the target tooth shape when the optical linewidth is zero. In fact, we can show that if we fix $\mathcal{L}(\omega)$, such convolution does not affect the optimality of the square shape.

Recall the efficiency functional in Eqn.~\ref{eqn:eff}. Now we want to replace $f(\omega)$ with $f(\omega)*\mathcal{L}(\omega)$:
\begin{align}
    \eta(L) =& \left\vert\frac{LT}{2\pi}\int_{-\pi/T}^{\pi/T}f(\omega)*\mathcal{L}(\omega)e^{i\omega T}d\omega\right\vert^2 \nonumber\\
    &\times e^{-\frac{LT}{2\pi}\int_{-\pi/T}^{\pi/T}f(\omega)*\mathcal{L}(\omega)d\omega} \nonumber\\
    \approx & \left\vert\frac{LT}{2\pi}\left(\int_{-\pi/T}^{\pi/T}\mathcal{L}(\omega)e^{i\omega T}d\omega\right)\left(\int_{-\pi/T}^{\pi/T}f(\omega)e^{i\omega T}d\omega\right)\right\vert^2 \nonumber\\
    &\times e^{-\frac{LT}{2\pi}\left(\int_{-\pi/T}^{\pi/T}\mathcal{L}(\omega)d\omega\right)\left(\int_{-\pi/T}^{\pi/T}f(\omega)d\omega\right)} \nonumber\\
    =& |F_{-1}L'|^2e^{-F_0L''},
\end{align}
where for the approximation we have used the convolution theorem to decouple the integrals, and assumed that the intrinsic line shape is not too wide. In the end, we have defined $L'=L\left\vert\int_{-\pi/T}^{\pi/T}\mathcal{L}(\omega)e^{i\omega T}d\omega\right\vert$ and $L''=L\int_{-\pi/T}^{\pi/T}\mathcal{L}(\omega)d\omega$, which are constants as long as the line shape $\mathcal{L}(\omega)$ is a fixed function, and thus do not affect the previous proof. Therefore, the square tooth shape is still the most desirable target tooth shape. Although the final observable tooth shape will be different from the perfect square shape, it still gives the highest achievable efficiency, as long as the intrinsic line shape does not vary.

We can re-write the efficiency as:
\begin{equation}
\begin{aligned}
    \eta(L) =& \left\vert\frac{L'}{L''}\right\vert^2|F_{-1}L''|^2e^{-F_0L''}\\
    \approx& \left\vert\int_{-\infty}^{\infty}\mathcal{L}(\omega)e^{i\omega T}d\omega\right\vert^2|F_{-1}L|^2e^{-F_0L},
\end{aligned}
\end{equation}
where we have used the approximate normalization condition $\int_{-\pi/T}^{\pi/T}\mathcal{L}(\omega)d\omega\approx\int_{-\infty}^{\infty}\mathcal{L}(\omega)d\omega=1$, and $\int_{-\pi/T}^{\pi/T}\mathcal{L}(\omega)e^{i\omega T}d\omega\approx\int_{-\infty}^{\infty}\mathcal{L}(\omega)e^{i\omega T}d\omega$, both assuming that the intrinsic line width is much smaller than the comb period $2\pi/T$. It is then obvious that, as long as the actual shape can be expressed as a convolution of the target shape and a fixed kernel, the effect of the finite optical linewidth is simply to scale the ideal efficiency by a multiplicative constant determined by the Fourier transform of the intrinsic line shape. We note that in practice the convolution kernel can have complicated form, but in general it will depend on optical coherence time of the atoms. In addition, the retrieval efficiency's dependence on the square-tooth width, comb period and optical depth is unchanged, which makes the optimal width of the square tooth unchanged as well, independent of the optical linewidth. We emphasize again that in experimental scenarios, after the optical depth and the desired comb period are determined, the corresponding optimized square-tooth width can be obtained analytically~\cite{bonarota2010efficiency}.

\section{Comparison with Lorentzian and Gaussian tooth shapes}\label{sec:comparison}
In practice, errors in the control of optical pumping may result in deviation from the optimal tooth shape as target. For instance, suppose we want to create the optimal square-tooth AFC under a certain optical depth constraint. The actual tooth width might be different from what we intend to create. Therefore, it is important to examine how robust and practical is the advantage of the square-tooth AFC in retrieval efficiency. Here we make a direct comparison among square, Lorentzian, and Gaussian shapes of AFC teeth by evaluating the achievable retrieval efficiency under different tooth widths and optical depths. Recall that we have shown that the inclusion of finite optical linewidth only scales the efficiency by a multiplicative constant determined by the intrinsic line shape. For the three tooth shapes considered in this section which are symmetric, the effect of finite background absorption is also just a reduction factor. Therefore, without loss of generality, we focus on the ideal tooth shapes.

\subsection{Retrieval efficiencies under different tooth widths and optical depths}
We still impose a physical constraint on maximal height of the comb shape $\alpha_M$. In this case, we consider Lorentzian and Gaussian line shapes with FWHM $\Gamma$ (then half FWHM is $\Gamma/2$) described by
\begin{align}
& L_{\alpha_M,\Gamma}(\omega) = \frac{\alpha_M\Gamma^2}{\Gamma^2 + 4\omega^2},\\ & G_{\alpha_M,\Gamma}(\omega) = \alpha_Me^{-4\ln 2\frac{\omega^2}{\Gamma^2}},
\end{align}
respectively. In order to obtain their corresponding retrieval efficiencies according to Equation~\ref{eqn:eff}, we need to evaluate two definite integrals which do not result in analytical expressions that consist of elementary functions, but involve non-elementary integrals such as error function, cosine and sine integrals. Therefore, we would like to evaluate the efficiencies numerically and thus we would like to make the functions dimensionless to get rid of the influence of storage time $T$ and crystal length $L$:
\begin{align}
    \eta\left[f(\omega)\right] =& \left(\frac{TL}{2\pi}\int_{-\pi/T}^{\pi/T}f(\omega)\cos(\omega T)dx\right)^2 e^{-\frac{TL}{2\pi}\int_{-\pi/T}^{\pi/T}f(\omega)d\omega}\nonumber\\
    =& \left(\frac{L}{2\pi}\int_{-\pi/T}^{\pi/T}f(\omega)\cos(\omega T)d(\omega T)\right)^2 e^{-\frac{L}{2\pi}\int_{-\pi/T}^{\pi/T}f(\omega)d(\omega T)}\nonumber\\
    =& \left(\frac{L}{2\pi}\int_{-\pi}^{\pi}\bar{f}(t)\cos(t)dt\right)^2 e^{-\frac{L}{2\pi}\int_{-\pi}^{\pi}\bar{f}(t)dt}
\end{align}
where $t\coloneqq \omega T$ and $\bar{f}(t)=\bar{f}(\omega T)=f(\omega)$, while we have also assumed symmetric comb shapes s.t. $f(\omega)=f(-\omega)$ which is satisfied by $L(\omega)$ and $G(\omega)$. Then for Lorentzian and Gaussian shapes the numerics-friendly expressions of retrieval efficiencies are
\begin{align}
    \eta_L(p,\mathrm{OD}) =& \left(\frac{1}{2\pi}\int_{-\pi}^{\pi}\frac{p^2\mathrm{OD}}{p^2+4t^2}\cos(t)dt\right)^2 e^{-\frac{1}{2\pi}\int_{-\pi}^{\pi}\frac{p^2\mathrm{OD}}{p^2+4t^2}dt},\\
    \eta_G(p,\mathrm{OD}) =& \left(\frac{1}{2\pi}\int_{-\pi}^{\pi}\mathrm{OD}e^{-4\ln 2\frac{t^2}{p^2}}\cos(t)dt\right)^2 e^{-\frac{1}{2\pi}\int_{-\pi}^{\pi}\mathrm{OD}e^{-4\ln 2\frac{t^2}{p^2}}dt},
\end{align}
respectively, where the effect of storage time $T$ is represented by a phase factor $p\coloneqq\Gamma T\in[0,2\pi]$ (dimensionless width which is proportional to the inverse finesse), while the effect of crystal length $L$ is represented by the maximum optical depth $\mathrm{OD}\coloneqq \alpha_ML\geq 0$. For the square tooth with FWHM $\Gamma$ and OD constraint $\alpha_ML$ the efficiency is~\cite{bonarota2010efficiency} (for review of other analytical properties of square-tooth AFC see Appendix~\ref{sec:afc_eff})
\begin{align}
    \eta_S(p,\mathrm{OD}) = \frac{\mathrm{OD}^2\sin^2(p)}{\pi^2}e^{-\frac{p\mathrm{OD}}{\pi}}.
\end{align}
The retrieval efficiencies for AFCs with square, Lorentzian and Gaussian teeth under different optical depths $\mathrm{OD}$ and dimensionless widths $p$ are visualized in Figure~\ref{fig:efficiency} to offer a more comprehensive view of the AFC performance under varying experimental conditions. There is a noteworthy feature of AFC retrieval efficiencies that emerges for different tooth shapes: As the maximum optical depth $\mathrm{OD}$ increases, the desired range of tooth widths decreases, outside which the retrieval efficiency will decrease quickly as the width deviates from the optimal width. Among the three typical tooth shapes considered here, the Lorentzian tooth is arguably the ``worst'', in that it achieves the lowest efficiency under a fixed $\mathrm{OD}$, while the desirable range of the tooth width is the narrowest. On the other hand, the advantage of the square tooth is conspicuous visually, which will be further elaborated in the following.

\begin{figure}[t]
    \centering
    \includegraphics[width=\linewidth]{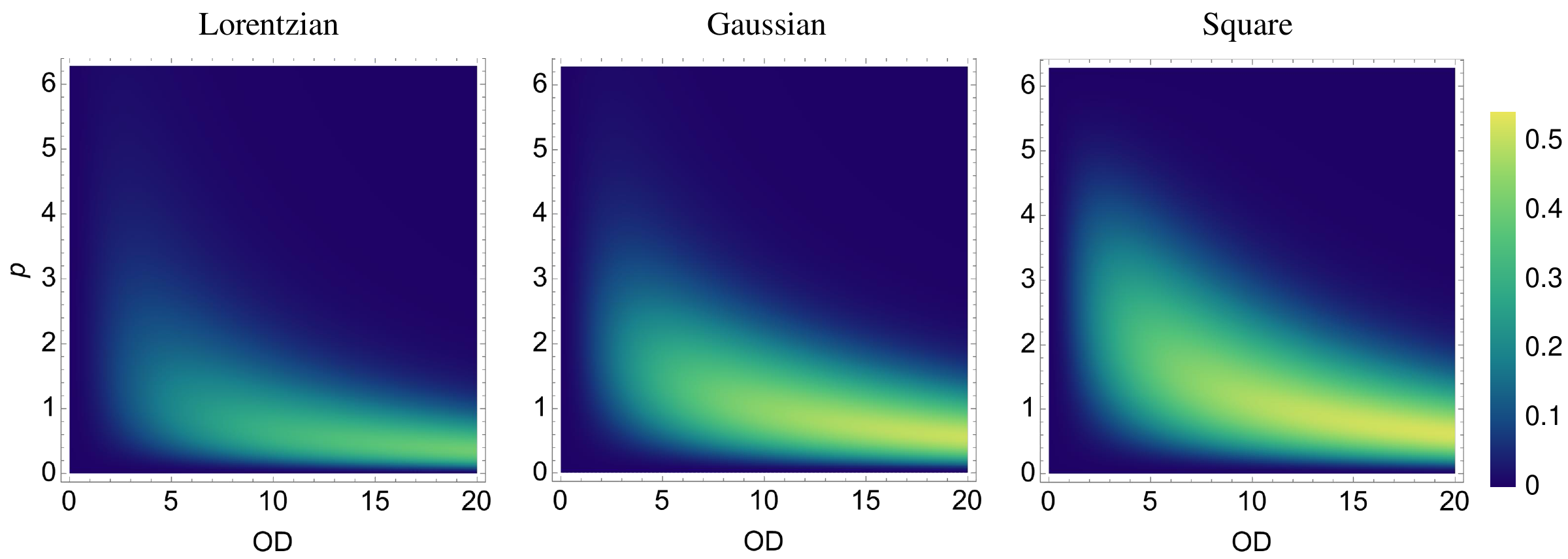}
    \caption{Retrieval efficiencies of AFC with Lorentzian, Gaussian, and square teeth, for different tooth widths $p=\Gamma T\in[0,2\pi]$ and optical depths $\mathrm{OD}=\alpha_ML\in[0,20]$. The color map is identical for all three figures, with the darkest blue color corresponding to zero efficiency and the brightest yellow color corresponding to $0.54$ which is the upper limit for forward retrieval efficiency of AFC memories.}
    \label{fig:efficiency}
\end{figure}

\subsection{Robust advantage of square tooth}
We demonstrate the robustness of the square-tooth AFC's advantage in retrieval efficiency. Specifically, we compare the retrieval efficiency of the square-tooth AFC with different widths and the optimal retrieval efficiency of AFCs with Lorentzian or Gaussian teeth, under fixed optical depths. The retrieval efficiencies of AFCs with Lorentzian or Gaussian teeth as functions of $p$ and $\mathrm{OD}$ do not have simple closed forms, but their maximum values given fixed $\mathrm{OD}$ can be determined numerically. 

We first examine the absolute difference between the retrieval efficiency of square-tooth AFC with different widths and the optimal retrieval efficiency of Lorentzian-tooth and Gaussian-tooth AFCs, under fixed maximum optical depth $\mathrm{OD}$. Specifically, we calculate the following quantity
\begin{align}
    D_{L/G}(p,\mathrm{OD}) = \eta_S(p,\mathrm{OD}) - \max_{p'\in[0,2\pi]}\eta_{L/G}(p',\mathrm{OD}).
\end{align}
We visualize $D_L(p,\mathrm{OD})$ and $D_G(p,\mathrm{OD})$ in the upper panel of Figure~\ref{fig:comparison}. We also consider the relative difference between the retrieval efficiency of square-tooth AFC and the optimal retrieval efficiency of Lorentzian-tooth and Gaussian-tooth AFCs, under fixed $\mathrm{OD}$. Specifically, we calculate the following quantity
\begin{align}
    R_{L/G}(p,\mathrm{OD}) = \frac{\eta_S(p,\mathrm{OD}) - \max_{p'\in[0,2\pi]}\eta_{L/G}(p',\mathrm{OD})}{\max_{p'\in[0,2\pi]}\eta_{L/G}(p',\mathrm{OD})}.
\end{align}
$R_L(p,\mathrm{OD})$ and $R_G(p,\mathrm{OD})$ are visualized in the lower panel of Figure~\ref{fig:comparison}, to complement $D_L(p,\mathrm{OD})$ and $D_G(p,\mathrm{OD})$. 

It is noteworthy that for the visualization we have normalized any negative value to zero, which corresponds to the darkest blue that is uniform in most areas in each subfigure. For each subfigure any value higher than the maximum value in the color bar is also normalized to the maximum for the color bar, corresponding to the brightest yellow, with the lower right panel for $R_G(p,\mathrm{OD})$ as an example. Therefore, in the region where the color is green or yellow we have that $D/R_{L/G}(p,\mathrm{OD})>0$, i.e. the retrieval efficiency of square-tooth AFC with relative tooth width $p$ is higher than the optimal retrieval efficiency that can be achieved by Lorentzian-tooth and Gaussian-tooth AFCs under the same maximum optical depth $\mathrm{OD}$. The existence of such regions means that in practice we do not have to create square teeth with optimal widths to demonstrate advantage over other practical tooth shapes such as Lorentzian and Gaussian using the same atomic ensemble (such that $\mathrm{OD}$ is unchanged), i.e. the advantage is robust. It can be observed that as the maximum optical depth $\mathrm{OD}$ increases the range of tooth widths which support advantage in retrieval efficiency decreases, which is justified by the feature of retrieval efficiencies. 

\begin{figure}[t]
    \centering
    \includegraphics[width=0.85\linewidth]{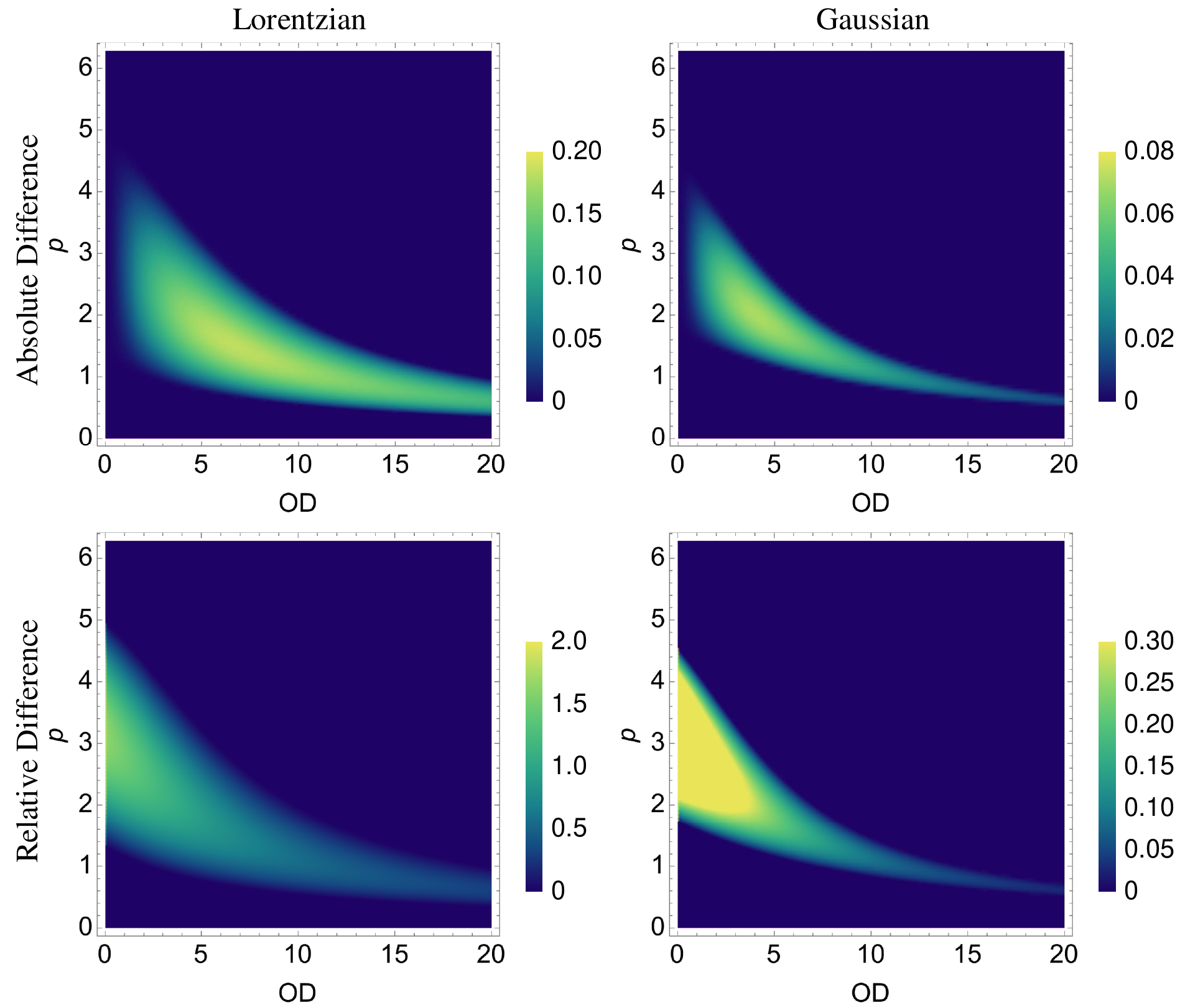}
    \caption{Comparison between the retrieval efficiency of square-tooth AFC and the optimal retrieval efficiency of Lorentzian-tooth and Gaussian-tooth AFC, under fixed maximum optical depth $\mathrm{OD}$. The upper panel demonstrates the absolute differences $D_L(p,\mathrm{OD})$ and $D_G(p,\mathrm{OD})$. The lower panel demonstrates the relative differences $R_L(p,\mathrm{OD})$ and $R_G(p,\mathrm{OD})$. All functions are visualized for $p\in[0,2\pi]$ and $\mathrm{OD}\in[0,20]$. Any negative value is normalized to zero, which corresponds to the darkest blue, while for each subplot any value higher than the maximum value of the color bar is also normalized to brightest yellow.}
    \label{fig:comparison}
\end{figure}

\section{Conclusion and discussion}\label{sec:conclusion}
We offer a rigorous proof which shows that the square tooth with an optimized width is the best tooth shape for AFC memories, based on a semi-classical model of AFC absorption and retrieval processes. It achieves the highest retrieval efficiency, even when finite background absorption and optical linewidth are considered. Although the ideal square-tooth AFC is not obtainable in practice, our results suggest that a target square-tooth comb will still lead to the highest retrieval efficiency. We further reinforce the realistic advantage of square-tooth AFC through explicit comparison with two other common tooth shapes, Lorentzian and Gaussian, which reveals the robustness of square-tooth AFC's retrieval efficiency advantage even when the tooth width deviates a bit from the optimal width.

This work demonstrates an interesting example of applying mathematical analysis techniques to quantum memory theory. We have also identified generalized functional optimization problems to which the proof technique can be immediately applied, as presented in Appendix~\ref{sec:generalization}. The techniques used in this work are expected to find use in or inspire mathematically rigorous studies of the optimal configuration for other physical scenarios, especially different quantum memory protocols which have higher efficiency upper bounds, such as backward retrieval of AFC and cavity-enhanced AFC~\cite{afzelius2010impedance,moiseev2010efficient}. Notably, cavity enhancement could potentially result in additional requirement on pulse engineering as imposed by the cavity mode, so the optimization of cavity-enhanced quantum memories still requires detailed analysis, which we leave for future work. Nevertheless, cavities whose mode line widths are larger than the AFC bandwidth can still be used for enhancing AFC memories~\cite{afzelius2010impedance} as demonstrated experimentally in ~\cite{zhong2017nanophotonic,craiciu2019nanophotonic}, and in such cases our results are still directly applicable. Moreover, it is interesting to take into account further realistic details in the quantum memory protocols. For instance, recent studies, e.g.~\cite{ortu2022multimode}, have discovered more subtle and non-trivial interplay between optical pumping schemes and the observed tooth shape. The interplay between the spectral width of the AFC structure and the spectral width of the signal fields could also limit the efficiency and accuracy of the AFC protocol in contrast to the CRIB protocol~\cite{moiseev2012rephasing,arslanov2017optimal,arslanov2019maps}. In addition, we would like to reemphasize that we have assumed ideal periodic comb in our analysis, which is not exactly the same in practice where the tooth height will change depending on the distance from the center of the inhomogeneous broadening spectrum. The impact of realistic aperiodicity in the comb can be an interesting topic to study from both mathematical and practical perspectives. We leave the exploration of these aspects for future works.

\section*{Data availability statement
}
No new data were created or analysed in this study.

\section*{Acknowledgments}
A.Z. would like to thank Yuzhou Chai for helpful discussions. This work is funded by the NSF Quantum Leap Challenge Institute for Hybrid Quantum Architectures and Networks (NSF Grant No. 2016136).

\appendix

\section{Detailed solution to the Maxwell-Bloch equations~\ref{eqn:MB_eqn}}\label{sec:solution}
For completeness and pedagogical benefits, we provide a detailed solution to the Maxwell-Bloch equations used in both~\cite{bonarota2010efficiency} and this work.

We consider that when the incoming signal pulse has a characteristic time $\tau$ that is much longer than the time it travels through the medium $L/c$, we can ignore the temporal retardation effects~\cite{afzelius2009multimode}, i.e. we can neglect the temporal derivative in the first equation in Eqn.~\ref{eqn:MB_eqn} which describes the field (Rabi frequency) dynamics. Then it is natural to work in the frequency domain by performing the Fourier transform on both sides of the equations, where $\tilde{\omega}$ denote the frequency domain argument to differentiate from the frequency detuning $\omega$. Subsequently the equations are reduced to:
\begin{equation}
\begin{aligned}
    \partial_z \Omega(z,\tilde{\omega}) =& -\frac{i}{2\pi}\int d\omega f(\omega)P(\omega;z,\tilde{\omega}),\\
    i\tilde{\omega} P(\omega;z,\tilde{\omega}) =& -(i\omega + \gamma)P(\omega;z,\tilde{\omega}) - i\Omega(z,\tilde{\omega}),
\end{aligned}
\end{equation}
where the Fourier transform $\Omega(z,\tilde{\omega})$ has a series expansion due to our previous retrieval pulse sequence ansatz:
\begin{equation}
\begin{aligned}
    \Omega(z,\tilde{\omega}) =& \mathcal{F}[\sum_{p\geq 0}a_p(z)\Omega(0,t-pT)]\\
    =& \sum_{p\geq 0}a_p(z)e^{-ip\tilde{\omega} T}\mathcal{F}[\Omega(0,t)]\\
    =& \Omega(0,\tilde{\omega})\sum_{p\geq 0}a_p(z)e^{-ip\tilde{\omega} T},
\end{aligned}
\end{equation}
where $\mathcal{F}[\cdot]$ denotes the Fourier transform from the time domain to the frequency domain, and $\Omega(0,\tilde{\omega})\coloneqq\mathcal{F}[\Omega(0,t)]$. The integral on the right hand side of the first equation can be expanded as:
\begin{equation}\label{eqn:delta_int}
\begin{aligned}
    &\int_{-\infty}^{\infty} d\omega f(\omega)P(\omega;z,\tilde{\omega})\\
    =& \sum_{n\leq 0}F_n \int_{-\infty}^{\infty} d\omega e^{-in\omega T}P(\omega;z,\tilde{\omega}),
\end{aligned}
\end{equation}
where the infinite lower limit of the integral is an approximation, as given a finite center resonance frequency, the lower bound on detuning is finite. From the second equation above we obtain the relation between the field and the polarization:
\begin{equation}
    P(\omega;z,\tilde{\omega}) = \frac{-1}{\tilde{\omega} + \omega - i\gamma}\Omega(z,\tilde{\omega}).
\end{equation}
Then the integral in Eqn.~\ref{eqn:delta_int} can be further written as:
\begin{equation}
\begin{aligned}
    &\sum_{n\leq 0}F_n \int_{-\infty}^{\infty} d\omega e^{-in\omega T}P(\omega;z,\tilde{\omega})\\
    =& -\Omega(z,\tilde{\omega})\sum_{n\leq 0}F_n \int_{-\infty}^{\infty} d\omega \frac{e^{-in\omega T}}{\tilde{\omega} + \omega - i\gamma}.
\end{aligned}
\end{equation}
One can notice that the integral on the right hand side is closely related to the Hilbert transform~\cite{kschischang2006hilbert,king2009hilbert} which is defined as
\begin{align}
    \mathcal{H}[u(t)]\coloneqq \frac{1}{\pi}\mathrm{pv}\int_{-\infty}^{\infty}d\tau \frac{u(\tau)}{t-\tau},
\end{align}
where $\mathrm{pv}$ denotes the Cauchy principal value. Besides the Hilbert transform, we note another widely-used identity
\begin{align}
    \frac{1}{x+i0^{\pm}}=\mathrm{pv}\frac{1}{x}\mp i\pi\delta(x).
\end{align}
To utilize this identity for deriving analytical formulae, we consider another approximation of small homogeneous broadening, i.e. $\gamma\rightarrow 0^+$~\cite{sonajalg1994diffraction}. Then we explicitly evaluate the aforementioned integral for two cases: $n=0$ and $n<0$. When $n=0$ we have:
\begin{equation}
\begin{aligned}
    &\int_{-\infty}^{\infty} d\omega \frac{1}{\tilde{\omega} + \omega - i\gamma}\\
    =& \mathrm{pv}\int_{-\infty}^{\infty} d\omega \frac{1}{\tilde{\omega} + \omega} + i\pi\int_{-\infty}^{\infty} d\omega \delta(\tilde{\omega} + \omega) = i\pi.
\end{aligned}
\end{equation}
When $n<0$ we have:
\begin{align}
    &\int_{-\infty}^{\infty} d\omega \frac{e^{-in\omega T}}{\tilde{\omega} + \omega - i\gamma}\nonumber\\
    =& \mathrm{pv}\int_{-\infty}^{\infty} d\omega \frac{e^{-in\omega T}}{\tilde{\omega} + \omega} + i\pi\int_{-\infty}^{\infty} d\omega e^{-in\omega T}\delta(\tilde{\omega} + \omega)\nonumber\\
    =& \pi\mathcal{H}[e^{-i(nT)\omega}] + i\pi e^{in\tilde{\omega} T} = 2i\pi e^{in\tilde{\omega} T}.
\end{align}
Equipped with the above results, we return to the Maxwell-Bloch equations which have been reduced to:
\begin{equation}
\begin{aligned}
    &\sum_{p\geq 0}\left[\partial_za_p(z)\right]e^{-ip\tilde{\omega} T} \\
    =& \sum_{p\geq 0}a_p(z)e^{-ip\tilde{\omega} T}\frac{i}{2\pi}\sum_{n\leq 0}F_n \int_{-\infty}^{\infty} d\omega \frac{e^{-in\omega T}}{\tilde{\omega} + \omega - i\gamma}.
\end{aligned}
\end{equation}

According to the definition of retrieval efficiency $\eta(L)\coloneqq|a_1(L)|^2$, to obtain $a_1(z)$ we only need to solve for two equations corresponding to $p=0,1$:
\begin{align}
    \partial_za_0(z) =& -\frac{1}{2}F_0a_0(z),\\
    \partial_za_1(z) =& -\frac{1}{2}F_0a_1(z) - F_{-1}a_0(z),
\end{align}
which must satisfy two boundary conditions: zero decay of input signal at $z=0$, i.e. $a_0(0)=1$; zero forward retrieval signal at $z=0$, i.e. $a_1(0)=0$. Then it is easy to derive their explicit expressions:
\begin{align}
    & a_0(z)=e^{-F_0z/2},\\
    & a_1(z)=-F_{-1}e^{-F_0z/2}z.
\end{align}

\section{Detailed proofs}\label{sec:proof_detail}
In this section we provide detailed proofs of propositions in the main text.

\subsection{Proof of Proposition 1}
\begin{proof}
First recall that the square tooth shape $f_s(\omega)$ takes value $\alpha_M$ on $[-\Gamma,\Gamma]$, and is zero on $[-\pi/T,-\Gamma]\cup[\Gamma,\pi/T]$. 

We consider an arbitrary symmetric shape $f(\omega)$ with the same area as the square tooth, which means that $f(\omega)$ satisfies:
\begin{align}\label{eqn:equal_area}
    \int_0^{\pi/T}f(\omega)d\omega = \int_0^{\Gamma}\alpha_Md\omega = \alpha_M\Gamma.
\end{align}
To prove the proposition, we would like to show that $f(\omega)$ cannot achieve higher retrieval efficiency than the square shape $f_s(\omega)$, i.e.
\begin{align}\label{eqn:higher_eff}
    \int_0^{\pi/T}f(\omega)\cos(\omega T)d\omega \leq \int_0^{\Gamma}\alpha_M\cos(\omega T)d\omega.
\end{align}
Note that in principle we need to account for $f(\omega)$ satisfying $\int_0^{\pi/T}f(\omega)\cos(\omega T)d\omega \leq 0$, which corresponds to the case where the tooth is more concentrated on $[\pi/2T,\pi/T]$. However, in such cases we can always redefine coordinate for one single tooth under periodic boundary condition, so that under the new coordinate $\int_0^{\pi/T}f(\omega)\cos(\omega T)d\omega\geq 0$, which we focus on without loss of generality.

For simplicity we define a function $\tilde{f}(\omega)$ as the difference between the function $f(\omega)$ and the square function $f_s(\omega)$:
\begin{equation}
\begin{aligned}
    \tilde{f}(\omega) &\coloneqq f(\omega) - f_s(\omega)\\
    &= 
    \begin{cases}
    f(\omega) - \alpha_M & \omega\in[0, \Gamma]\\
    f(\omega) & \omega\in[\Gamma, \pi/T]
    \end{cases},
\end{aligned}
\end{equation}
which satisfies $\int_0^{\pi/T}\tilde{f}(\omega)d\omega=0$ according to the assumption of identical area. The proof is then divided two scenarios: (i) $\Gamma\leq\pi/2T$, and (ii) $\Gamma>\pi/2T$.

In scenario (i) we have:
\begin{align}
    &\int_0^{\pi/T}\tilde{f}(\omega)\cos(\omega T)d\omega \nonumber\\
    =& \int_0^{\Gamma}\tilde{f}(\omega)\cos(\omega T)d\omega + \int_{\Gamma}^{\pi/T}\tilde{f}(\omega)\cos(\omega T)d\omega \nonumber\\
    \leq& \cos(\Gamma T)\int_0^{\Gamma}\tilde{f}(\omega)d\omega + \cos(\Gamma T)\int_{\Gamma}^{\pi/T}\tilde{f}(\omega)d\omega \nonumber\\
    =& \cos(\Gamma T)\int_0^{\pi/T}\tilde{f}(\omega)d\omega = 0,
\end{align}
where for the inequality we used the following facts:
\begin{align}
    & \tilde{f}(\omega)\leq 0,~\omega\in[0,\Gamma],\\
    & \cos(\omega T)\geq\cos(\Gamma T)\geq 0,~\omega\in[0,\Gamma],\\
    & \cos(\omega T)\leq\cos(\Gamma T),~\omega\in[\Gamma,\pi/2T],
\end{align}
Then we have:
\begin{equation}
    \int_0^{\pi/T}f(\omega)\cos(\omega T)d\omega \leq \int_0^{\Gamma}\alpha_M\cos(\omega T)d\omega.
\end{equation}

In scenario (ii), we have:
\begin{equation}
\begin{aligned}
    &\int_0^{\pi/T}\tilde{f}(\omega)\cos(\omega T)d\omega\\
    =& \int_0^{\pi/2T}\tilde{f}(\omega)\cos(\omega T)d\omega + \int_{\pi/2T}^{\Gamma}\tilde{f}(\omega)\cos(\omega T)d\omega\\
    &+ \int_{\Gamma}^{\pi/T}\tilde{f}(\omega)\cos(\omega T)d\omega\\
    \leq& 0 + \cos(\Gamma T)\int_{\pi/2T}^{\Gamma}\tilde{f}(\omega)d\omega + \cos(\Gamma T)\int_{\Gamma}^{\pi/T}\tilde{f}(\omega)d\omega \leq 0.\label{eqn:scenario2_bound}
\end{aligned}
\end{equation}
In the above, for the first inequality we have considered the following facts. 
\begin{align}
    \tilde{f}(\omega)\leq 0,~\omega\in[0,\pi/2T],
\end{align}
because $f(\omega)\leq \alpha_M$ when $\omega\in[0,\Gamma]$, while $\Gamma>\pi/2T$ in scenario (ii). This gives the upper bound for the first term in Eqn.~\ref{eqn:scenario2_bound}.
\begin{align}
    \cos(\Gamma T)\tilde{f}(\omega)\geq\cos(\omega T)\tilde{f}(\omega)\geq 0,~ \omega\in[\pi/2T,\Gamma],
\end{align}
because on this interval we have $0\geq\cos(\omega T)\geq\cos(\Gamma T)$ while again $\tilde{f}(\omega)\leq 0$. This gives the upper bound for the second term. 
\begin{align}
    0\geq\cos(\Gamma T)\tilde{f}(\omega)\geq\cos(\omega T)\tilde{f}(\omega),~\omega\in[\Gamma,\pi/T],
\end{align}
because $\tilde{f}(\omega)\geq 0$ while $0\geq\cos(\Gamma T)\geq\cos(\omega T)$ on this interval, which gives the upper bound for the third term. 

For the second inequality we have used:
\begin{align}
    & \cos(\Gamma T)<0,~\Gamma>\pi/2T,\\
    & \int_{\pi/2T}^{\pi/T}\tilde{f}(\omega)d\omega = \int_0^{\pi/T}\tilde{f}(\omega)d\omega - \int_0^{\pi/2T}\tilde{f}(\omega)d\omega\nonumber\\
    &= 0 - \int_0^{\pi/2T}\tilde{f}(\omega)d\omega \geq 0.
\end{align}
In the end we have:
\begin{equation}
    \int_0^{\pi/T}f(\omega)\cos(\omega T)d\omega \leq \int_0^{\Gamma}\alpha_M\cos(\omega T)d\omega.
\end{equation}

The proposition is thus proved.
\end{proof}

\subsection{Proof of Lemma 2}
\begin{proof}
The family of square shapes can be explicitly formulated as a piecewise function:
\begin{equation}
    f_{(\Gamma,c)}(\omega) = 
    \begin{cases}
    0 & \text{if } \omega\in[-\frac{\pi}{T}, c-\Gamma]\cup[c+\Gamma, \frac{\pi}{T}]\\
    \alpha_M & \text{if } \omega\in[c-\Gamma, c+\Gamma]
    \end{cases}.
\end{equation}
Then the objective integral can be evaluated directly as:
\begin{equation}
    I(\Gamma,c) = \alpha_M\int_{c-\Gamma}^{c+\Gamma}\cos(\omega T)d\omega = \frac{2\alpha_M\sin(\Gamma T)\cos(cT)}{T}
\end{equation}
And we can evaluate its partial derivative against the center $c$ to get:
\begin{equation}
    \frac{\partial}{\partial c}I(\Gamma,c) = -\frac{2\alpha_M\sin(\Gamma T)\sin(cT)}{T}
\end{equation}
which takes zero value at $c=0$ and is negative for all $c\in[0,\pi/T-\Gamma]$, positive for all $c\in[-\pi/T+\Gamma,0]$, as $0\leq\Gamma\leq\pi/T$ results in $\sin(\Gamma T)\geq 0$.
Therefore we conclude that for a fixed half width $\Gamma$, for $-\pi/T+\Gamma\leq c\leq\pi/T-\Gamma$ the integral has the maximal value when $c=0$.
\end{proof}

\subsection{Proof of Proposition 3}
\begin{proof}
Recall that we have expressed $|F_{-1}|$ as
\begin{align}
    |F_{-1}| = \frac{T}{2\pi}\left\vert\int_{-\pi/T}^{\pi/T}f'(\omega')\cos(\omega'T)d\omega'\right\vert,
\end{align}
and the redefined $f'(\omega')$ is still a valid shape function. Moreover, after the redefinition, $f'(\omega')$ will satisfy the following condition that $\int_{-\pi/T}^{0}f'(\omega')\sin(\omega' T)d\omega' = -\int_{0}^{\pi/T}f'(\omega')\sin(\omega' T)d\omega'$. Therefore, in the following we will focus on $f'(\omega')$, and denote it as $f(\omega)$ for simplicity. 

For $f(\omega)$ s.t. $\int_{-\pi/T}^{\pi/T}f(\omega)d\omega=S$, we can denote their integral on $[0,\pi/T]$ and $[-\pi/T,0]$ as $S_r$ and $S_l$, respectively, s.t. $S = S_r+S_l$. Then we define two square shapes on either side of $x=0$ with the areas being $S_r$ and $S_l$:
\begin{equation}
\begin{aligned}
    &f_l(\omega)=
    \begin{cases}
    0 & \text{if } \omega\in[-\frac{\pi}{T}, -\frac{S_l}{\alpha_M}]\cup[0, \frac{\pi}{T}]\\
    \alpha_M & \text{if } \omega\in[-\frac{S_l}{\alpha_M}, 0]
    \end{cases},\\ 
    &f_r(\omega)=
    \begin{cases}
    0 & \text{if } \omega\in[-\frac{\pi}{T}, 0]\cup[\frac{S_r}{\alpha_M}, \frac{\pi}{T}]\\
    \alpha_M & \text{if } \omega\in[0, \frac{S_r}{\alpha_M}]
    \end{cases},
\end{aligned}
\end{equation}
so that $f_s(\omega)\coloneqq f_r(\omega)+f_l(\omega)$ has an identical area as the tooth shape $f(\omega)$. $|F_{-1}|$ can be easily upper bounded by separating the $\omega\geq 0$ and $\omega\leq 0$ parts, respectively:
\begin{align}
    |F_{-1}| =& \frac{T}{2\pi}\left\vert\int_{-\pi/T}^{\pi/T}f(\omega)\cos(\omega T)d\omega\right\vert \nonumber\\
    \leq& \frac{T}{2\pi}\left\vert\int_{-\pi/T}^{0}f(\omega)\cos(\omega T)d\omega\right\vert \nonumber\\
    &+ \frac{T}{2\pi}\left\vert\int_{0}^{\pi/T}f(\omega)\cos(\omega T)d\omega\right\vert.
\end{align}

Then according to the result of Proposition 1, we have that:
\begin{align}
    & \left\vert\int_{-\pi/T}^{0}f(\omega)\cos(\omega T)d\omega\right\vert \leq \int_{-\pi/T}^{0}f_l(\omega)\cos(\omega T)d\omega,\\
    & \left\vert\int_{0}^{\pi/T}f(\omega)\cos(\omega T)d\omega\right\vert \leq \int_{0}^{\pi/T}f_r(\omega)\cos(\omega T)d\omega,
\end{align}
which leads to
\begin{align}
    |F_{-1}| \leq \frac{T}{2\pi}\int_{-\pi/T}^{\pi/T}f_s(\omega)\cos(\omega T)d\omega.
\end{align}
Moreover, according to Lemma 2 we have 
\begin{equation}
\begin{aligned}
    &\frac{T}{2\pi}\int_{-\pi/T}^{\pi/T}f_s(\omega)\cos(\omega T)d\omega\\
    \leq& \frac{T}{2\pi}\int_{-\pi/T}^{\pi/T}f_{\Gamma=S/(2\alpha_M),c=0}(\omega)\cos(\omega T)d\omega    
\end{aligned}
\end{equation}
where the upper bound is achieved by a square tooth. Therefore, there does not exist a tooth shape that can achieve a higher retrieval efficiency than the square tooth with the fixed area.
\end{proof}

\section{Properties of the square-tooth AFC retrieval efficiency}\label{sec:afc_eff}
We have shown that the inclusion of finite optical linewidth only scales the efficiency by a multiplicative constant determined by the intrinsic line shape. Therefore, without loss of generality, we focus on the ideal square-tooth AFC. For completeness and ease of reference, in this appendix we present some properties of the square-tooth AFC retrieval efficiency as shown in~\cite{bonarota2010efficiency}.

For the ideal square-tooth AFC with a tooth half width $\Gamma$ and height constraint $\alpha_M$ (i.e. optical depth $\mathrm{OD}=\alpha_ML$) the efficiency can be derived as:
\begin{equation}
    \eta_S(\Gamma, \mathrm{OD}) = \frac{1}{\pi^2}\mathrm{OD}^2\sin^2(\Gamma T)e^{-\frac{\Gamma T}{\pi}\mathrm{OD}},
\end{equation}
which gives the optimal tooth half width for the square tooth:
\begin{equation}\label{eqn:opt_square}
    \Gamma_\mathrm{opt}^S(\mathrm{OD}) = \frac{1}{T}\arctan\left({\frac{2\pi}{\mathrm{OD}}}\right).
\end{equation}

We see that the ratio between $\Gamma_\mathrm{opt}^\mathrm{S}$ and a quarter of the comb period $\pi/2T$ satisfies:
\begin{equation}
    0 \leq \frac{2}{\pi}\arctan\left({\frac{2\pi}{\mathrm{OD}}}\right) \leq 1,
\end{equation}
for any positive optical depth $\mathrm{OD}$. That is, the optimal width of the square tooth will never extend beyond $\pi/2T$, which is intuitive both physically and mathematically. Physically, if the comb is too wide the frequency difference between the emitters within one comb is large which will lead to dephasing, an undesirable outcome that will harm the transition dipole rephasing. Mathematically, $\cos x$ is above zero on $[0,\pi/2T]$ while below zero on $[\pi/2T,\pi/T]$, therefore if $f(\omega)\geq 0$ on $[\pi/2T,\pi/T]$ the shape will actually result in a lower value of the first integral in the definition of $F_{-1}$ (and for the square tooth which is symmetric with respect to $\omega=0$ the second integral in $F_{-1}$'s definition is simply zero). 

The retrieval efficiency of the optimized square-tooth AFC as function of the optical depth is:
\begin{equation}
    \eta_{S,\mathrm{opt}}(\mathrm{OD}) = \frac{4e^{-\frac{\mathrm{OD}}{\pi}\arctan\frac{2\pi}{\mathrm{OD}}}}{1 + \frac{4\pi^2}{\mathrm{OD}^2}}.
\end{equation}
It can be easily proved that the above function increases monotonically as the optical depth increases. Then by taking the limit of the retrieval efficiency as the optical depth approaches infinity, we can find its maximal value:
\begin{equation}
    \eta_\mathrm{max} = \lim_{\mathrm{OD}\rightarrow+\infty}\eta_{S,\mathrm{opt}}(\mathrm{OD}) = \frac{4}{e^2} \approx 54.1\%,
\end{equation}
which agrees with the upper bound derived in~\cite{afzelius2009multimode}.

\section{Proof review and generalization}\label{sec:generalization}
The optimization of the AFC tooth shape is a functional optimization problem. The retrieval efficiency is a functional of the tooth shape function, which in general requires calculus of variations or numerical techniques. Our approach only utilizes elementary analytical techniques, and only basic properties of the objective functional are relevant. Therefore, the proof procedures can be easily generalized. Here we review the proof from a high-level perspective, and discuss the form of the functionals to which the techniques can be readily applied.

First, recall that the argument of the objective functional (the AFC retrieval efficiency in the above example) is required to be bounded and non-negative. Second, the functional has a structure which can be decomposed as a product of two sub-functionals: (i)~the first sub-functional is a monotonically increasing function of the inner product of the argument function $f(x)$ (the AFC shape function in the above example), and another basis function $g(x)$ (sinusoidal function for Fourier series in the above example), $\int_a^bf(x)g(x)dx$, and (ii)~ the second sub-functional is an arbitrary function of $\int_a^bf(x)dx$. In fact, it is the product structure that allows us to simplify the proof by focusing on the maximization of $\int_a^bf(x)g(x)dx$. For the proof that the square function maximizes $\int_a^bf(x)g(x)dx$ under the assumed constraints of positivity and boundedness, we have only utilized the boundedness and monotonicity of the basis function $g(x)$.

According to the above review, we arrive at the generalized statement about the optimality of the square-function as follows:
\begin{theorem}
    Consider real-valued functionals in the following form:
    \begin{equation}
        F[f(x)] = G\left(\int_a^bf(x)g(x)dx\right)H\left(\int_a^bf(x)dx\right),
    \end{equation}
    where the real-valued function $G(x)$ is monotonically increasing and the function $H(x)$ is also real-valued; the real-valued function $g(x)$ is bounded and monotonically decreasing on $x\in[a,b]$, while the real-valued argument function satisfies $0\leq f(x)\leq \alpha$, s.t. $f(x)$ and $f(x)g(x)$ are both integrable on $x\in[a,b]$. The square function
    \begin{equation}
        f_s(x)=
        \begin{cases}
        0 & x\in[c, b]\\
        \alpha & x\in[a, c]
        \end{cases},\\
    \end{equation}
    with optimized width $(c-a)$ will achieve the maximal value of $F[f(x)]$.
\end{theorem}

\section*{References}
\bibliographystyle{unsrt}
\bibliography{references}

\end{document}